\documentclass{article}
\usepackage[utf8]{inputenc}
\usepackage[T1]{fontenc}
\usepackage[dvipsnames]{xcolor}
\usepackage{url}
\usepackage{booktabs}
\usepackage{amsfonts}
\usepackage{amsmath}
\usepackage{amssymb}
\usepackage{nicefrac}
\usepackage{tikz}
\usepackage{bm}
\usepackage{pgfplots}
\usepgfplotslibrary{patchplots}
\pgfplotsset{compat=1.18}
\usepackage{microtype}
\usepackage{tcolorbox}

\usepackage{graphicx}
\usepackage{algorithm}
\usepackage[noEnd=true]{algpseudocodex}
\algrenewcommand\algorithmicrequire{\textbf{Input:}}
\algrenewcommand\algorithmicensure{\textbf{Output:}}
\usepackage[switch]{lineno}
\usepackage{dirtytalk}
\usepackage{mathtools}
\usepackage{amsthm}
\usepackage{authblk}
\usepackage[a4paper, total={6.5in, 8.5in}]{geometry}

\usepackage{thmtools}
\usepackage{thm-restate}

\declaretheorem[name=Proposition]{proposition}

\usepackage[colorlinks = true, linkcolor=NavyBlue,citecolor=ForestGreen]{hyperref}
\usepackage[nameinlink]{cleveref}

\newtheorem{theorem}{Theorem}

\newtheorem{definition}{Definition}

\newfloat{procedure}{htbp}{lop}
\floatname{procedure}{Procedure}

\usepackage{complexity}

\title{Leveraging Matchings in Constrained Fair Division with a Conflict Graph}
\date{}

\author[1,2]{Evangelos Markakis}

\author[1]{Michalis Samaris}

\affil[1]{Athens University of Economics and Business, Greece}
\affil[2]{Archimedes/Athena RC, Greece}

\begin{document}
\maketitle

\begin{abstract}
We study the problem of allocating indivisible goods under constraints, expressed via a conflict graph $G$. In such an instance, the $m$ items are the vertices of $G$ and connected items cannot be allocated in the same bundle. Under this model, it is already known that EF1 allocations may not exist. Our main contribution is an analysis parametrized by the maximum degree $\Delta(G)=\Delta$ on the existence and computation of complete EF1 allocations. We address this question in various cases by leveraging results from matching theory. First, we provide a tight existence result for agents with ordered valuations and for the broader class of tiered valuations. We present an algorithm that returns an EF1 allocation when then number of items does not exceed a specific bound. This bound is determined by $n$ and $\Delta$, and it is tight when $\Delta$ is greater than $2n/3$. We also construct an approximation algorithm when $m$ exceeds this bound. For general additive valuations the problem becomes more challenging. Given the current impossibility results, we focus on the case where the number of items is at most $2n$. For this case, we provide an almost complete picture for the instances that admit EF1 allocations, by combining Round Robin with matchings.
\end{abstract}

\section{Introduction}

Fair Division with indivisible items has attracted significant research interest in the recent years. 
This can be evidenced by the plethora of new results within the last decade, as exhibited in surveys such as \cite{DBLP:journals/ai/AmanatidisABFLMVW23}, but also by the emergence of motivating applications, including course allocation algorithms \cite{7c65302b-f079-361a-94f1-0c3c9f6fc76b}, and food donation programs \cite{Mertzanidis0V24}, among others.
The main question in this context is to derive fair allocations of the available resources and several fairness criteria have been proposed, as a way to alleviate the fact that more traditional fairness notions, such as envy-freeness and proportionality fail to exist. 

Among the various relaxations of envy-freeness that have been suggested, our work focuses on the well-known EF1 criterion (\textit{envy-freeness up to one good}) defined by \cite{7c65302b-f079-361a-94f1-0c3c9f6fc76b}.
An allocation is said to be EF1 if for any pair of agents $i, j$, for which agent $i$ envies the bundle of agent $j$, the envy is eliminated after the hypothetical removal of some item from the bundle of $j$. 
This is a concept that can be easily satisfied whenever there are no further constraints on the allocation space. Namely, both a simple round-robin algorithm, as well as the  envy cycle elimination algorithm by \cite{lipton2004approximately} are guaranteed to produce EF1 assignments.


The status regarding the existence of EF1 allocations however changes drastically when there are additional restrictions on the feasible assignments. Recently, there have been various works that study fair division problems under constraints \cite{10.1145/3505156.3505162}, motivated by different application scenarios.
Indicatively, some works examine the division of items into categories, with a constraint on the number of items of each category allocated to each agent \cite{biswas2018fair}. 
In another model, 
the items form the vertices of a graph and each bundle must form a connected subgraph \cite{ijcai2017p20,BILO2022197}. 
We refer to our related work section for further exposition.

Our work focuses on yet another graph-theoretic model of constrained fair division that has been initiated recently, where the constraints are modeled via a \textit{conflict graph}. More precisely, the items correspond to the vertex set of the conflict graph and a bundle is not allowed to contain two items that are connected with an edge. 
In other words, the bundles allocated to the $n$ agents must constitute a $n$-coloring in the conflict graph (with each bundle being an independent set). 
This model turns out to impose severe limitations on the existence of fair allocations, if we want to allocate all items (i.e., have a complete allocation). 
Firstly, a complete allocation may simply not exist, even without imposing any fairness criterion. To circumvent this, the mild condition of $\Delta<n$, where $\Delta$ is the maximum degree of the graph does guarantee the existence of a complete allocation (since this is sufficient for the existence of a valid $n$-coloring). 
Even under this assumption, counterexamples for the existence of EF1 allocations have already been identified \cite{Hummel2021,ijcai2025p435}.
We feel this naturally poses the question of understanding further or characterizing when is it possible to achieve EF1 or other fairness concepts under such constraints.


\subsection{Our contribution}
We are interested in the existence and efficient computation of complete, EF1 allocations in the presence of a conflict graph, when the agents have additive valuations. 
The main contribution of our work is that we explore the effect of the maximum degree of the conflict graph, $\Delta$, on the existence of EF1 allocations. For the cases where we establish existence, we also describe efficient algorithms for computing such allocations.
The techniques we use are based on leveraging results from matching theory such as Hall's Theorem for perfect matchings or the structure of vertex sets that form violations to Hall's theorem when no perfect matching exists.
We derive our results by combining these techniques with algorithmic tools of fair division such as Round Robin and envy cycle elimination procedures. 

We start in Section \ref{section: ordered}, where we deal with ordered valuations and the wider class of tiered valuations. Ordered valuations, where the agents agree on the ranking of all items w.r.t. their value has been extensively considered in fair division (e.g. \cite{10.1137/19M124397X}). Tiered valuations form a more natural superclass, where the items are grouped into tiers, and the agents agree on the ranking between the tiers even if they may not agree on the ranking of the items within each tier. 
We consider valuations with tiers of size $n$, meaning that each tier (except maybe the last one) contains exactly $n$ items (so that agents agree on which are the $n$ most valuable items, then the next $n$ most valuable and so on).
For this valuation class we provide in Section \ref{subsection: MRR} an algorithm that returns an EF1 allocation when the number of items, $m$, is bounded as: $m\leq n\left\lfloor \frac{n}{\Delta} \right\rfloor +n-\Delta$.
We show that this result is tight, i.e., when $m$ exceeds the above bound there may not exist EF1 allocations, and this is proved by extending the counterexample of Igarashi et al.  \cite{ijcai2025p435}. 
To alleviate this negative fact when $m> n\left\lfloor \frac{n}{\Delta} \right\rfloor +n-\Delta$, we provide a positive approximation result in Section \ref{subsection: MRR-approach}. Namely, for $\Delta\leq \frac{n}{2}$, and based on our previous algorithm, we provide an approximation algorithm obtaining a $(1-\frac{2}{r+1})$-EF1 allocation where $r=\left\lfloor \frac{n}{\Delta} \right\rfloor$.
Observe that, the ratio of the algorithm approaches $1$ 
as $r$ becomes sufficient small and tends to $0$.

In Section \ref{section: Additive} we move to general additive valuations. The setting here is much more restrictive, as we exhibit that even with identical valuations, an EF1 allocation may not exist when $m> 2n - \Delta$. We therefore focus on understanding EF1 existence when $m\leq 2n$, for which we provide an almost complete picture.
Firstly, we combine the Round Robin algorithm with an appropriately constructed matching to derive an EF1 allocation when either $m\leq 2n-\Delta$ or $m\leq 2n$ and $\Delta\leq\frac{n}{2}$.
Our most technical result is that we extend this algorithm to attain EF1 when $m\leq 2n$ and $\Delta\leq \frac{2n}{3}$. As a complementary result, when $\Delta>\frac{2n}{3}$, we can attain EF1, when $m\leq 2n-\Delta+1$, and $G$ does not have $K_{\Delta,2n-2\Delta+1}$ as a subgraph. The last 2 results are obtained by a repeated application of matching algorithms, in combination with envy cycle removals. 

\subsection{Further related work}
The model of fair division with a conflict graph, that we study here was first introduced by Chiarelli et al. \cite{10.1007/978-3-030-48966-3_12}, where they focus on a different objective, that of maximizing the lowest total profit of items allocated to any agent. 
The existence and computation of complete EF1 allocations under this model, which is also the focus of our work, was later studied by Hummel and Hetland \cite{Hummel2021}, where they established both positive results for some special cases, as well as impossibility results for general additive valuations (e.g., non-existence when $\Delta> n$ or when $n\geq 4$ and with lower values of $\Delta$). 
In the same work they also considered the criteria of MMS and Maximum Nash Welfare under this setting.
Igarashi et al. \cite{ijcai2025p435} proved that a maximal but not necessarily complete EF1 allocation always exists for two agents and any graph, when the agents have monotone valuations. Moreover, they showed NP-hardness results for finding maximal EF1 allocations for three or more agents.
A special case of this problem, that corresponds to interval scheduling was considered by Li et al. \cite{NEURIPS2021_b5b1d9ad} and Kumar et al. \cite{10.5555/3635637.3663155}, under the fairness criteria of EF1 and MMS. Additionally, Kumar et al. \cite{10.5555/3635637.3663155} also considered the case where the items are chores.
Finally, Yoneda and Yoneda \cite{yoneda2026fairdivisionsoftconflicts} studied a setting where the conflict constraints are treated as "soft." In this scenario, allocating adjacent goods to the same agent is allowed and the objective is to find EF1 allocations with a small number of such conflict violations.

There have been several other models of fair division under constraints. 
Biswas et al. \cite{biswas_et_al:LIPIcs.FORC.2023.8} considered a model where items are grouped into categories and each bundle needs to respect a bound on the items received from each
category.
Another graph-related model involves the allocation of connected bundles, where items are located again in a graph, as studied in \cite{ijcai2017p20,BILO2022197}.
Other models concern constraints defined via matroids. For example, \cite{biswas2018fair} considered a setting where the constraints are determined by laminar matroids. Furthermore, Gourves et al. \cite{GMT13} considered a model where the allocated items should form a base of a matroid.
Additional results under matroid constraints regarding MMS and other criteria have also been obtained in \cite{GMT14,gourves2019maximin}.
Yet another model focuses on the scenario where all agents should receive the same number of items (e.g. in assigning shifts) studied in  \cite{FGM14} and more recently in \cite{BOGOMOLNAIA2025,MMP25}.
For an overview of problems involving constrained fair division we refer to the survey by Suksompong \cite{10.1145/3505156.3505162} or to the one by Biswas et al. \cite{Biswas2023}.

When it comes to the unconstrained setting, there is an even bigger volume of research activity within the last decades. 
The lack of envy-free allocations led over the years to various new criteria including the one we study here, defined by \cite{7c65302b-f079-361a-94f1-0c3c9f6fc76b}, as well as the notions of EFX (envyfreeness up to any good) introduced in \cite{GMT14,10.1145/3355902}, MMS defined by \cite{7c65302b-f079-361a-94f1-0c3c9f6fc76b}, and groupwise MMS considered in \cite{barman2018groupwise}. For an overview, we refer to the survey by \cite{DBLP:journals/ai/AmanatidisABFLMVW23}.

\section{Preliminaries}\label{section: Preliminaries}

We consider a setting with $n$ agents and $m$ indivisible items to be allocated to the agents.
We denote by $N$ the set of agents and by $M$ the set of items. 
Each agent $i\in N$ has an \textit{additive} valuation function $v_i$ over the items so that for every $S\subseteq M$, $v_i(S)=\sum_{g\in S} v_i(\{g\})$ and $v_i(\emptyset)=0$. For simplicity, we will use $v_i(g)$ instead of $v_i(\{g\})$, and we also assume that $v_i(g)\geq 0$ for every $g\in M$. We denote by $V$ the set of valuation functions of the $n$ agents, i.e., $V=\{v_i\mid i\in N\}$.

A \textit{complete allocation}, or simply an allocation, denoted by $A=(A_1,A_2,\ldots$, $A_n)$, is a partition of the items $M$, where $\bigcup_{i\in [n]}A_i=M$ and $A_i\cap A_{i^{\prime}}=\emptyset$ for any $i,i^{\prime}\in[n]$ with $i\neq i^{\prime}$. 
The bundle $A_i$ is the set of items assigned to agent $i$.
A \textit{partial allocation} is an allocation over a strict subset of the set of items $M$.

In the constrained setting we study, certain pairs of items are not allowed to be allocated to the same agent. 
These pairs are represented via an undirected \textit{conflict graph} $G=(M,E)$, where its vertices are the set of items i.e.,  $V(G)=M$ and for every edge $(u,v)\in E$, the items $u$ and $v$ are conflicting and cannot be assigned to the same bundle. 
An allocation (complete or partial) is 
\textit{feasible} if no bundle contains a pair of conflicting items.
In other words, a feasible allocation forms a $n$-coloring of $G$ and each bundle $A_i$ is an independent set of $G$. 
The question of interest in our work is to find feasible and complete allocations which also satisfy some fairness criteria.

\begin{definition}
    We denote by $(N,M,V,G)$ an instance of fair allocation with a conflict graph, where
    \begin{itemize}
        \item[-] $N$ is the set of $n$ agents,
        \item[-] $M$ is the set of $m$ items,
        \item[-] $V$ is a set of $n$ additive valuation functions, with $v_i$ being the valuation function of agent $i$,
        \item[-] $G=(M,E)$ is the conflict graph.

    \end{itemize}
\end{definition}

Our analysis will make use of the maximum degree $\Delta(G)$, of a vertex in $G$. For brevity, we will denote the maximum degree by $\Delta$, instead of $\Delta(G)$. We also denote by $N(v)$ the neighbor vertices of a vertex $v$. For a subset of items $S\subseteq V(G)$, we use $N(S)$ for the vertices that are neighbors with at least one vertex belonging to $S$ i.e., $N(S)=\{u\in V \setminus S: (u,v)\in E \text{ for some } v\in S\}$. 
Given a partial allocation $A=(A_1,A_2,\ldots,A_n)$, the items that have as neighbor in $G$ an item allocated to agent $i$ are the \textit{conflicting items of agent $i$} w.r.t. $A$, i.e., these are precisely the items of the set $C_i(A)= N(A_i)$. 

\noindent {\bf Feasibility graph.} We define now an additional graph that will be helpful in formalizing some of our algorithms and their analysis. Given a partial allocation $A$, and the set $U$ of unallocated items in $A$, the \textit{feasibility graph}, denoted as $F(A,U)$, is a bipartite graph where $V(F(A,U))=N\cup U$ and $E(F(A,U))=\{(i,g)\mid g\notin C_i(A)\}$. Hence, in $F(A,U)$  each agent $i$ is connected with the items $g\in U$ such that $A_i\cup \{g\}$ remains feasible.

In the sequel, we will often allocate items to agents based on matchings that we compute in the feasibility graph. For brevity we will be using the procedure below. In particular, given a partial allocation $A$, and a subset $R$ of edges of $F(A,U)$ such that each $g\in U$ belongs to at most one edge of $R$, Procedure~\ref{allocating-edges} allocates for each $(i,g)\in R$ the item $g$ to agent $i$. 
\begin{procedure}
\caption{Allocate($A,R$)}
\label{allocating-edges}
\begin{algorithmic}[1]

\Require Partial allocation $A=(A_1,\dots, A_n)$, $R\subseteq E(F(A,U))$ s.t. every $g$ belongs to at most one edge of $R$.
\For{each $(i,g)\in R$}
    \State $A_i=A_i\cup\{g\}$
\EndFor

\end{algorithmic}
\end{procedure}


Beyond completeness and feasibility, we want allocations that also meet some fairness criterion. 
The fairness criterion in our work is \textit{envy-freeness up to one item} (EF1). 
Given an allocation $A$ we say that an agent $i$ \textit{envies} agent $j$ if $v_i(A_i)<v_i(A_j)$.
An allocation is \textit{envy-free} (EF) if no agent envies the bundle of some other agent i.e., $v_i(A_i)\geq v_i(A_j)$ for every pair of agents $i, j\in N$.
Envy-freeness up to one item is a relaxation of envy-freeness where every agent $i$ is not envious towards any other agent $j$ after the hypothetical removal of some item from the bundle of $j$.
The formal definition of an EF1 allocation  is following.

\begin{definition}
In an instance $(N,M,V,G)$, an allocation $A$ is envy-free up to one item (EF1) if, for every pair of agents $i, j \in N,$ it holds that $v_i(A_i)\geq v_i(A_j \setminus {g})$ for some $g\in A_j$.
\end{definition}

Similarly to the above definition, an allocation is said to be $\alpha-$EF1 if $v_i(A_i)\geq \alpha v_i(A_j \setminus {g})$ for some $g\in A_j$.

Due to the following known proposition, and since we ask for complete matchings, we only consider instances with $\Delta<n$.
\begin{proposition}\cite{Hummel2021}
    For any conflict graph $G=(M,E)$ with $\Delta \geq n$ there is an instance $(N,M,V,G)$ that has no EF1 allocation.
\end{proposition}


\subsection{Definitions and tools from matching theory} 
Now, we present some preliminary concepts related to matchings, which will be used extensively in the following sections.

In a bipartite graph $G = (X\cup Y,E)$  a \textit{perfect matching} is a matching that matches every vertex in $X$ to a unique vertex in $Y$ and vice versa. For a perfect matching to exist in a bipartite graph, it is necessary that $|X|=|Y|$.
If $|X|\leq |Y|$ we say that a matching \textit{saturates} $X$ if every vertex in $X$ is matched to a distinct vertex in $Y$. In the sequel, we will extensively use the graph theoretic version of Hall's theorem as follows.

\begin{theorem}[Hall's Theorem, \cite{hall1935}]\label{Hall-Thm}
    A bipartite graph  $G = (X\cup Y,E)$ with $|X|\leq |Y|$ has a matching that saturates all vertices in $X$ if and only if for every subset $ S \subseteq X$, it holds that  $|N(S)| \geq |S|$.
\end{theorem}

Clearly, when $|X|=|Y|$ and Hall's condition holds, the matching is perfect.
In a bipartite graph $G = (X\cup Y,E)$, we say that $(S,N(S))$ is a \textit{Hall violator} if $|S|>|N(S)|$. 
Finally, a known consequence of Hall's Theorem is the following proposition. For the sake of completeness, we provide a proof in the Appendix.

\begin{restatable}{proposition}{HallCor}
\label{Hall-Cor}
    If in a bipartite graph $G=(X\cup Y, E)$, where $|Y|=n^{\prime}$ and $|X|\leq n^{\prime}$, each vertex in $Y$ has degree at least $d$ and each vertex in $X$ has degree at least $n^{\prime}-d$, then there exists a matching saturating every vertex of $X$ in $G$.
\end{restatable}

\section{Ordered Valuations and Beyond}\label{section: ordered}

We say that the agents have \textit{ordered valuations} if there is an ordering of the items, say $g_1,g_2,\ldots,g_m$, such that $v_i(g_1)\geq v_i(g_2)\geq \ldots \geq v_i(g_m)$ for every $i\in N.$ Ordered valuations are motivated by settings where there is a common perception for the ranking of the items, even if the agents do not agree about their precise  value. The results we present in this section actually hold for a broader class, that we term \textit{valuations with tiers of size $n$}. 
Let $k$ be the minimum integer such that $(k-1)n< m\leq kn$.
When the agents have valuation with (common) tiers of size $n$, the set of items $M$ is partitioned into sets (tiers) $M_1, M_2,\ldots,M_k$.
The set $M_1$ contains the top $n$ most preferred items which are common to all agents, even though they may not agree on the ranking of these items. The set $M_2$ contains the next $n$ items, and so on. Only $M_k$ may contain strictly less than $n$ items.
A formal definition follows.
\begin{definition}
\label{def:tiered}
    The agents have valuations with tiers of size $n$ if there is a partition of the items in tiers $M_1, M_2,\ldots,M_k$, such that
    \begin{itemize}
        \item[(i)] $|M_j|=n$ for every $j\in[k-1]$ and $|M_k|\leq n$.
        \item[(ii)] For any agent $i$ and for every items $g\in M_j$, $g^{\prime}\in M_{j^{\prime}}$ where $j<j^{\prime}$ it holds that $v_i(g)\geq v_i(g^{\prime})$.
    \end{itemize}
\end{definition}

Ordered valuations form a special case of the above definition, where the items belonging to the same tier $M_j$ are also ranked in the same way by all agents.
The following simple proposition  holds for valuations with tiers of size $n$ and thus also for ordered valuations.

\begin{restatable}{proposition}{balancedEFONE} \label{balanced_EF1}
    Consider an instance as per Definition \ref{def:tiered}, with tiers $M_1, M_2,\ldots$, $M_k$. Every allocation $A=(A_1,A_2,\ldots, A_n)$ such that  $|A_i\cap M_j|=1$ for every $i\in [n]$, $j\in[k-1]$, and $|M_k\cap A_i|\leq 1$, is EF1.
\end{restatable}

\begin{proof}
For two agents $i, i^{\prime}$, let  $A_i$ and $A_{i^{\prime}}$ be their bundles and let $A_i\cap M_j=\{g_{i}^{j}\} $ and $A_{i^{\prime}}\cap M_j=\{g_{i^{\prime}}^{j}\}$ for every $j\in [k-1]$. If $|A_{i^{\prime}}\cap M_k|=1$ let $g_{i^{\prime}}^{k}$ be the respective item of $A_{i^{\prime}}\cap M_k$.
Assume that $i$ envies $i^{\prime}$. It holds that $v_i(g_{i}^j)\geq v_i(g_{i^{\prime}}^{j+1})$ for every $j\in [k-2]$. Moreover, $v_i(g_{i}^{k-1})\geq v_i(g_{i^{\prime}}^{k})$ if $|A_{i^{\prime}}\cap M_k|=1$ or $v_i(g_{i}^{k-1})\geq 0$ if $|A_{i^{\prime}}\cap M_k|=0$.
Regardless of the value of $|A_{i^{\prime}}\cap M_k|$, summing the $k-1$ respective inequalities side by side in each case implies that $v_i(A_i\setminus M_k)\geq v_i(A_{i^{\prime}}\setminus{g_{i^{\prime}}^{1}})$. Thus,  whether or not $i$ has an item from $M_k$ in her bundle, it follows that $v_i(A_i)\geq v_i(A_{i^{\prime}}\setminus{g_{i^{\prime}}^{1}})$ and consequently $A$ is EF1.

\end{proof}

\subsection{EF1 via Tiered Matching}\label{subsection: MRR}

Our main result in this subsection is an algorithm that returns a complete EF1 allocation when the number of items is bounded by $n\left\lfloor \frac{n}{\Delta} \right\rfloor +n-\Delta$.

Our proposed algorithm iteratively allocates the items to the agents via matchings (one per tier) in the feasibility graph.
We briefly discuss first for how many agents an unallocated item is feasible with their bundles. 
For a partial allocation $A$ and a set of unallocated items $U$, an agent $i$ is not connected with an item $g\in U$ if some item which is a neighbor of $g$ in the conflict graph belongs to $A_i$, i.e., $g\in C_i(A)$. 
Since the conflict graph $G$ has maximum degree $\Delta$,
even if all of the at most $\Delta$ neighbor items of $g$ have already been given to $\Delta$ different agents, there are at least $n-\Delta$ agents whose bundle is feasible with item $g$. Thus, every $g\in U$ is connected with at least $n-\Delta$ agents in $F(A,U)$.
We are now in a position to present the \textit{Tiered Matching Algorithm} (Algorithm~\ref{MRR}).

\begin{algorithm}[!ht]
\caption{Tiered Matching}
\label{MRR}
\begin{algorithmic}[1] 

\Require $(N,M,V,G)$, with valuations having tiers of size $n$

\State Set $A_i=\emptyset$ for every agent $i$. 
\State Let $M_1,M_2,\ldots, M_k$ be the partition of items into the tiers.
\For{each $j$ from $1$ to $k$}
    \State Construct $F(A,M_j)$.
    \State Find a maximum size matching $R_j$ in $F(A,M_j)$.\label{matching-j}
    \State Allocate($A,R_j$)
\EndFor

\State \textbf{return} $A=(A_1,A_2,\ldots,A_n)$

\end{algorithmic}
\end{algorithm}

Because of Proposition~\ref{balanced_EF1}, if in each round of Algorithm~\ref{MRR}, the matching $R$ is perfect (except the last one, where it needs to saturate $M_k$), then the resulting allocation is complete and EF1.
In what follows, we prove that this condition holds when the number of items does not exceed $n\left\lfloor \frac{n}{\Delta} \right\rfloor +n-\Delta$. 

\begin{theorem}\label{MRR-Thm}
    Algorithm~\ref{MRR} runs in polynomial time and returns a complete EF1 allocation when $m\leq n\left\lfloor \frac{n}{\Delta} \right\rfloor +n-\Delta$ and the agents have valuations with tiers of size $n$.
\end{theorem}

\begin{proof}
We first assume that $m\leq n\left\lfloor \frac{n}{\Delta} \right\rfloor$. 
Then Algorithm~\ref{MRR} would run for at most $\left\lfloor \frac{n}{\Delta} \right\rfloor$ rounds. 
We prove that it returns a perfect matching in each of at most $\left\lfloor \frac{n}{\Delta} \right\rfloor-1$ rounds except the last one and it returns a matching saturating every item of $M_k$ in the last round.
In other words a matching saturating every item in each of at most $\left\lfloor \frac{n}{\Delta} \right\rfloor$ rounds.
We proceed by induction. In round 1, since $C_i(A)=\emptyset$ for every agent $i\in N$, we have $F(A,M_1)=K_{n,n}$. 
Hence, a perfect matching $R_1$ always exists. 
Now, let $A=(A_1,A_2,\ldots,A_n)$ be a partial allocation obtained by Algorithm~\ref{MRR} after round $r$ where in each round $j$, with $j\leq r$, the respective matching $R_j$ is perfect and $r\leq \left\lfloor \frac{n}{\Delta} \right\rfloor -1 $. 
In the partial allocation $A$, every agent $i$ has exactly $r$ items in her bundle $A_i$.
Therefore, since each item is represented by a vertex of the graph $G$ whose degree is at most $\Delta$, the conflict set $C_i(A)$ for each agent $i$ has at most $r\Delta\leq (\left\lfloor \frac{n}{\Delta} \right\rfloor -1)\Delta= \left\lfloor \frac{n}{\Delta} \right\rfloor \Delta- \Delta\leq n-\Delta$ items.
In round $r+1$, since $|C_i(A)|\leq n-\Delta$ and $|M_{r+1}|=n$ it is clear that $|M_{r+1}\cap C_i(A)|\leq n-\Delta$, and thus every agent $i$ has degree at least $\Delta$ in $F(A,M_{r+1})$. 
Combining this with the fact that each item has degree at least $n-\Delta$, it follows from Proposition~\ref{Hall-Cor} that a  matching saturating every vertex of $M_{r+1}$ exists in $F(A,M_{r+1})$.
By Proposition~\ref{balanced_EF1}, this concludes that if $m\leq n\left\lfloor \frac{n}{\Delta} \right\rfloor$, Algorithm~\ref{MRR} returns a complete EF1 allocation.

Now, let $m=n\left\lfloor \frac{n}{\Delta} \right\rfloor+ l$ where $1\leq l \leq n-\Delta$.
Each of the $l$ items belonging to $M_{k^{\prime}}$ where $k^{\prime}=\left\lfloor \frac{n}{\Delta} \right\rfloor+1$ is feasible with at least $n-\Delta$ agents and $l\leq n-\Delta$.
Therefore, in $F(A,M_{k^{\prime}})$ for every $S\subseteq M_{k^{\prime}}$, it holds that $|N(S)|\geq n-\Delta\geq |S|$.
Thus by Theorem~\ref{Hall-Thm} a matching saturating every vertex of $M_{k^{\prime}}$ exists and by Proposition~\ref{balanced_EF1} this concludes that Algorithm~\ref{MRR} returns a complete EF1 allocation when $m\leq n\left\lfloor \frac{n}{\Delta} \right\rfloor +n-\Delta$. 
\end{proof}

\noindent {\bf Tightness.} Igarashi et al. \cite{ijcai2025p435} have shown a counterexample to EF1 existence for any $n\geq 4$ with identical additive valuations and $m=n+2$.
This is not in conflict with our results since the conflict graph in their example is $K_{3,n-1}$ i.e., $\Delta(K_{3,n-1})=n-1$. Indeed,
for $\Delta=n-1$, $n\left\lfloor \frac{n}{\Delta} \right\rfloor=n$ and Theorem \ref{MRR-Thm} guarantees that Algorithm~\ref{MRR} returns an EF1 allocation when $m\leq n+n-\Delta=2n-(n-1)=n+1$. Therefore, since $m=n+2$ in the counterexample by Igarashi et al. \cite{ijcai2025p435}, this shows that our bound on $m$, for which Algorithm~\ref{MRR} returns an EF1 allocation is tight. 
In the following proposition we extend the tightness of Algorithm~\ref{MRR} w.r.t. the number of items, for any  $\Delta>\frac{2} {3}n$. In fact, we  strengthen this even further, and prove an impossibility result even for approximate EF1 allocations, beyond a certain constant ratio, when $m$ exceeds the bound of Theorem \ref{MRR-Thm}.

\begin{proposition}\label{counter-example}
    For any $n\geq 4$ and $\Delta$, such that $\frac{2n+1}{3}\leq \Delta\leq n-1$, there is an instance with identical valuations and $m=n\left\lfloor \frac{n}{\Delta} \right\rfloor +n-\Delta+1$ such that no complete $\alpha$-EF1 allocation exists for any $\alpha >\frac{\sqrt{2}}{2}$.
\end{proposition}

\begin{proof}
    Since $\frac{2n+1}{3}\leq \Delta\leq n-1$, we have $\left\lfloor \frac{n}{\Delta} \right\rfloor=1$.
    The counterexample is $K_{\Delta,2n-2\Delta+1}$. Let $V_1$ be the part with $\Delta$ vertices and $V_2$ the part with $2n-2\Delta+1$ vertices.
    Observe that $m=\Delta+2n-2\Delta+1=2n-\Delta+1=\left\lfloor \frac{n}{\Delta} \right\rfloor n +n-\Delta+1$. 
    The agents have the same valuation $v$, where $$v(j)=\begin{cases}
\sqrt{2}, &  j \in V_1 \\
1, &  j \in V_2
\end{cases}$$

Let $A=(A_1,A_2,\ldots, A_n)$ be a complete allocation. Since the conflict graph is a complete bipartite graph, it is clear that for every agent $i$ either $A_i\subseteq V_1$ or $A_i\subseteq V_2$ in any feasible allocation.

If the items of $V_1$ are allocated to $\Delta$ agents by giving one item of value $\sqrt{2}$ to each one, then the items of $V_2$ will be allocated among the remaining $n-\Delta$ agents.  
By the pigeonhole’s principle there will be some agent $i$ who receives at least three items of value $1$ among the $2n-2\Delta+1$ items of $V_2$ i.e., $v(A_i)\geq3$. 
Thus, even if one item $g$ of $A_i$ is removed the agents who are allocated one item of $V_1$ still envy agent $i$. 
Moreover, for any agent $j$ of the $\Delta$ agents receiving one item from $V_1$, it holds that $\frac{v(A_{j})}{v(A_i\setminus \{g\})}\leq\frac{\sqrt{2}}{2}$. Hence the allocation $A$ can be at most a $\frac{\sqrt{2}}{2}$-EF1 allocation.

The other case is that strictly less that $\Delta$ agents are allocated the items of $V_1$. 
Then, at least one agent $i$ would receive at least two items of value $\sqrt{2}$, so that $v(A_i)\geq 2\sqrt{2}$.
The items of $V_2$ would be allocated to at least $n-\Delta+1$ agents. 
Due to the pigeonhole’s principle there will be one agent $j$ that would receive at most one item among the $2n-2\Delta+1$ items of $V_2$. 
Thus, $v(A_{j})\leq 1$ and $j$ still envies $i$ even if one item of $A_i$ is removed. 
Moreover, it holds that $\frac{v(A_{j})}{v(A_i\setminus \{g\})}\leq\frac{1}{\sqrt{2}}=\frac{\sqrt{2}}{2}$.

Therefore, in both cases there is no EF1 allocation and additionally there always exist two agents, $i$ and $j$, such that $\frac{v(A_j)}{v(A_i\setminus \{g\})}\leq\frac{\sqrt{2}}{2}$. This implies that there is no complete $\alpha$-EF1 allocation for any $\alpha >\frac{\sqrt{2}}{2}$ in this instance. 

\end{proof}

Note that for the construction of Proposition \ref{counter-example}, we must have $\Delta\geq \frac{2n+1}{3}$. Otherwise, if $\Delta\leq\frac{2}{3}n$, then the degree of vertices in $V_1$ would be $2n-2\Delta+1\geq \frac{2}{3}n+1$ and the maximum degree of the conflict graph $K_{\Delta,2n-2\Delta+1}$ would not be equal to $\Delta$.
Thus, the counterexample no longer holds when $\Delta\leq\frac{2}{3}n$.

\subsection{Approximation via Tiered Matching}\label{subsection: MRR-approach}

Given the tightness of the previous result when $\Delta > 2n/3$, in this subsection, we investigate further the behavior of Algorithm~\ref{MRR} when there is no restriction on the number of items and for lower values of $\Delta$. More precisely, we explore the case where $\Delta\leq\frac{n}{2}$ (with the agents still having valuations with tiers of size $n$).
Note we can still run Algorithm \ref{MRR}, but for $j>\left\lfloor \frac{n}{\Delta} \right\rfloor$, there is no guarantee that the matching $R_j$ produced in line \ref{matching-j} of Algorithm~\ref{MRR} will be perfect.
If it is not perfect, this implies that there is a Hall violator $(S, N(S))$ in $F(A,M_{j})$ where $S$ is a subset of agents and $|S|>|N(S)|$. Since every item has degree at least $n-\Delta$ in $F(A,M_{j})$, for every subset of agents with size $\Delta+1$ or greater, each item is connected with at least one agent in the subset. Hence that would mean that $|N(S)| = n \geq |S|$. 
Therefore, this implies that $|S|\leq \Delta$.
Consequently $|N(S)|\leq \Delta-1$ and an agent $i$ that is unmatched in $R_j$ has at least $n-(\Delta-1)=n-\Delta+1$ items of $M_j$ in her conflict set.
Let $r=\left\lfloor \frac{n}{\Delta} \right\rfloor$ and $l$ be the minimum integer such that $m\leq lrn$.
In this section, if $(l-1)rn<m<lrn$, by adding zero valued items (without conflicts) we assume that $m=lrn$, so that Algorithm~\ref{MRR} will run for $lr$ rounds.

The following lemma shows a lower bound for the number of items each agent $i$ has in her bundle $A_i$ in an allocation $A$ obtained by Algorithm~\ref{MRR}.
For every $h\in [l]$ we denote by $M^h$ the items of the first $hr$ rounds i.e., $M^h= \bigcup_{j=1}^{hr}M_j$.

\begin{restatable}{lemma}{boundunlemma}
\label{bound-un}
    When $\Delta< n/2$, then for every  $h\in [l]$ and every agent $i\in N$, if $A_i$ is the bundle given to $i$ by Algorithm~\ref{MRR}, it holds that $|M^h \cap A_i|\geq h(r-1)+1$.
\end{restatable}

\begin{proof}
    Assume not and that for some $h\leq l$ an agent $i$ has at most $hr-h$ items in her bundle after $hr$ rounds of Algorithm~\ref{MRR}. This means that $i$ is unmatched in at least $h$ rounds of the algorithm.
    For every round $j$ that $i$ is unmatched in $R_j$ it holds that $|C_i(A)\cap M_j|\geq n-\Delta+1$.
    Hence, in total we have that $|C_i(A)|\geq h(n-\Delta+1)$.
    Moreover, since $|A_i|\leq (hr-h)$ then $|C_i(A)|\leq (hr-h)\Delta$. This implies that \begin{gather*}
    (hr-h)\Delta \geq h(n-\Delta+1)\\
     (r-1)\Delta\geq (n-\Delta+1)\\
     r\Delta\geq n+1\\
    \left\lfloor \frac{n}{\Delta} \right\rfloor \Delta\geq n+1
    \end{gather*}
     This leads to a contradiction and this concludes that  for every $h\leq l$ every agent has at least $hr-h+1$ items in her bundle after $hr$ rounds of Algorithm~\ref{MRR}.   
\end{proof}

Let $U$ be the set of unallocated items after the execution of Algorithm~\ref{MRR}. For every $h$ such that $2\leq h\leq l$, let $U_h$ be the subset of $U$ obtained from rounds $j$ such that $r(h-1)+1\leq j \leq rh$, i.e., $U_h=(\bigcup_{r(h-1)+1}^{rh}M_j)\cap U$.

The next algorithm is based on obtaining first an allocation $A$ by Algorithm~\ref{MRR} (which may not be complete) and then figuring out how to allocate the items of $U$. For every $U_h$, let $N_h$ be the set of agents who have less items from $M^h$ than the number of rounds $rh$ so far. Algorithm~\ref{Unalloc-Items} gives priority to the agents of $N_h$ for the unallocated items of $U_h$.
Then, if there are still items in $U_h$ and no such agent is feasible with some of them, these items are allocated to the other agents ensuring that every agent has at most $h(r+1)-1$ items from $M^h$ in her bundle.

\begin{algorithm}[!ht]
\caption{}
\label{Unalloc-Items}
\begin{algorithmic}[1] 

\Require $(N,M,V,G)$

\State $A = \text{Tiered Matching}(N,M,V,G)$. \label{line:RR}
\For{$h=2$ \textbf{to} $l$}
    \State Set $N_h=\{i\mid |A_i \cap M^h|< rh\}$.
    \While{$U_h\neq \emptyset$ and $\exists g\in U_h$ and $i\in N_h$ s.t. $g\notin C_i(A)$} \label{while1-start}
        \State $A_i=A_i\cup \{g\}$ and remove $g$ from $U_h$.
        \If{$|M^h\cap A_i|=rh$}
            \State Remove $i$ from $N_h$.
        \EndIf
    \EndWhile \label{while1-end}
    
    \State Set $N^{\prime}_h=\emptyset$. 
    \While{$U_h\neq \emptyset$ and $\exists g\in U_h$ and $i\notin N^{\prime}_h$ s.t. $g\notin C_i(A)$} \label{while2-start}
        \State $A_i=A_i\cup \{g\}$ and remove $g$ from $U_h$.
        \If{$|M^h\cap A_i|=h(r+1)-1$} \label{line:check}
            \State Add $i$ to $N^{\prime}_h$.
        \EndIf
    \EndWhile \label{while2-end}
\EndFor

\State \textbf{return} $A=(A_1,A_2,\ldots,A_n)$

\end{algorithmic}
\end{algorithm}

\begin{restatable}{lemma}{unallocatedlemma}
\label{unallocated}
When $\Delta \leq n/2$, the allocation $A$ returned by Algorithm~\ref{Unalloc-Items} is complete  and $h(r-1)+1\leq |M^h \cap A_i|\leq h(r+1)-1$  for every agent $i$ and every $h\leq l$.
\end{restatable}

\begin{proof}
    For the allocation $A$ as obtained by  Algorithm~\ref{MRR} in line \ref{line:RR} of Algorithm~\ref{Unalloc-Items} it holds that $h(r-1)+1<|M^h\cap A_i|\leq hr$ for every $h\in[l]$ by Lemma~\ref{bound-un}.
    Observe that the remaining of Algorithm~\ref{Unalloc-Items} does not remove items from any agent's bundle that are already allocated by Algorithm~\ref{MRR}. Therefore, $h(r-1)+1<|M^h\cap A_i|$ also holds for the returned allocation.
    Moreover, it is clear from the description of the algorithm (line \ref{line:check}) that if $|M^h\cap A_i|=h(r+1)-1$ for some $h\in [l]$ and some agent $i$, then no more items of $\bigcup_{j=1}^{hr}M_j$ are allocated to $A_i$. 
    This implies that $|M^h\cap A_i|\leq h(r+1)-1$ for every $h\in [l]$ and  $i\in N$.
    It remains to show that the returned allocation $A$ of  Algorithm~\ref{Unalloc-Items} is complete. 
     The items of $M_j$ for $j=1,2,\ldots, r$ have already been allocated in $A$.
     For every $h\geq 2$, after the while loop in lines \ref{while1-start}-\ref{while1-end}, if $U_h\neq \emptyset$, we claim that at least $n-\Delta$ agents have $hr$ items in their bundle or equivalently $|N_h|\leq \Delta$.
     Otherwise, if $|N_h|\geq \Delta+1$, since every $g\in U_h$ is feasible with at least $n-\Delta$ agents, then there is at least one agent $i\in N_h$ where $g\notin C_i(A)$ which contradicts the assumption that the while loop in lines \ref{while1-start}-\ref{while1-end} has terminated.
     Up to that point, every agent has at most $(h-1)(r+1)-1$ items from $M^{h-1}$ and at most $r$ from $M^h\setminus M^{h-1}$. Thus, overall at most $(h-1)(r+1)-1+r=h(r+1)-2$.
     At each verification of the while loop condition in line \ref{while2-start}, we will show that if $U_h\neq \emptyset$ then $|N^{\prime}_h|=n_2\leq \Delta-1$.
     The agents in $N_h$ have at least $h(r-1)+1$ where $|N_h|=n_1\leq \Delta$ and the agents neither in  $N_h$ nor in $N^{\prime}_h$ have at least $hr$ items. 
     Since $|M^h|=nhr$, then
      \begin{gather*}
      nhr\geq |\bigcup_{i\in N}(M^h \cap A_i)|= \\
      |\bigcup_{i\in N_h}(M^h \cap A_i)|+|\bigcup_{i\in N^{\prime}_h}(M^h \cap A_i)|+\\+|\bigcup_{i\notin (N_h\cup N^{\prime}_h)}(M^h \cap A_i)|\geq  \\
       n_1(h(r-1)+1) +n_2(h(r+1)-1)+(n-n_1-n_2)hr= \\
       n_1h+n_1+n_2h-n_2+nhr= \\
      nhr+(n_2-n_1)(h-1) 
      \end{gather*}
      
      If $n_2\geq \Delta$, as $n_1\geq \Delta$, observe that the above inequality can only hold with equality when $n_1=n_2=\Delta$ for every $h\geq 2$. 
      Therefore, if in some verification of the condition of the while loop in line \ref{while2-start} $|N^{\prime}_h|=\Delta$ then all the items of $M_h$ have already been allocated and $U_h=\emptyset$.
      Thus, as long as $U_h\neq \emptyset$, then $|N^{\prime}_h|\leq \Delta-1$.

      For every $g$ that remains in $U_h$ in the while loop of lines \ref{while2-start}-\ref{while2-end} from its at least $n-\Delta$ feasible agents, at most $\Delta-1$ of them belong to $N^{\prime}_h$.
      Therefore, there are remaining $n-\Delta-(\Delta-1)=n-2\Delta+1\geq 1$ (recall $\Delta\leq n/2$), that are feasible and not belonging to $N^{\prime}_h$. 
      This completes the proof, as every item in $U_h$ can be allocated to some agents and the allocation $A$ returned by Algorithm~\ref{Unalloc-Items} is complete. 
    \end{proof}

By Lemmas~\ref{bound-un} and~\ref{unallocated} the best and the worst case for a bundle $A_i$ of an agent $i$ are as represented in Table~\ref{bundles}.
Observe that in $A_b$ it holds that $|M^h\cap A_b|=h(r+1)-1$ for every $h\in [l]$.
The additional items of $A_b$ arise from the best possible rounds. For every $h\geq 2$ there are two items of round $M_{(h-1)r+1}$ in the bundle of $A_b$. 
This implies that the next possible round that $A_b$ may contain an extra item is $M_{hr+1}$, where there is also one extra item, because of the inequality $|M^h\cap A_b|\leq h(r+1)-1$ and so on.
Respectively, for $A_w$ it holds that $|M^h\cap A_w|=h(r-1)+1$ for very $h\in [l]$ and the lost items of $A_w$ occur in the worst possible way satisfying the inequality $|M^h\cap A_w|\geq h(r-1)+1$.


\begin{table*}[tbh]
\centering

\begin{tabular}{l|llll|llll|l|llll}
 & $M_1$ & $M_2$ & $\ldots$ & $M_r$ & $M_{r+1}$  &$M_{r+2}$  & $\ldots$ & $M_{2r}$  & $\ldots$ & $M_{(l-1)r+1}$ & $M_{(l-1)r+2}$ & $\ldots$ & $M_{lr}$    \\ \hline
$A_{b}$ & $1$ & $1$ & $1$ & $1$ & $2$ & $1$ & $1$  & $1$ & $\ldots$ & $2$ & $1$ & $1$ &  $1$    \\ \hline
$A_{w}$ & $1$ & $1$ & $1$ & $1$ & $0$ & $1$ & $1$  & $1$ & $\ldots$ & $0$ & $1$ & $1$ &  $1$    \\ \hline
\end{tabular}
\caption{The best and the worst case for the bundles $A_i$. The entries in the columns correspond to $|A_b\cap M_j|$ and $|A_w\cap M_j|$ for the respective column $M_j$. }
\label{bundles}
\end{table*}

\begin{theorem}
   When $\Delta \leq n/2$ and the agents have valuations with tiers of size $n$,  Algorithm~\ref{Unalloc-Items} runs in polynomial time and returns an allocation $A$ which is a $(1-\frac{2}{r+1})$-EF1 where $r=\left\lfloor \frac{n}{\Delta} \right\rfloor$. 
\end{theorem}

\begin{proof}
The best and the worst case for a bundle $A_i$ of allocation $A$  returned by Algorithm~\ref{Unalloc-Items}, denoted as $A_b$ and $A_w$ respectively, are shown in Table~\ref{bundles}. 
We will show that $v_i(A_w)\geq (1-\frac{2}{r+1})v_i(A_b\setminus \{\alpha\})$ where $A_b\cap M_1=\alpha$ and $i$ is the agent that $A_w$ is allocated to. 
Let $A_{w}^{h}=\bigcup_{j=(h-1)r+1}^{hr}A_w\cap M_j$  for every $h\in[l]$.

Because of the ordered valuations, for every $h\in[l-1]$ it holds that $v_i(A_b\cap M_{hr+1})|\leq \frac{2}{r-1}v_i(A_{w}^{h}))$ since every one of $r-1$ items of $A_{w}^{h}$ ($r$ if $h=1$) has at least as much value as every of two items of $A_b\cap M_{hr+1}$.
Observe that if we remove the items of $A_b\cap M_{hr+1}$ from $A_b$ for every $h\in [l-1]$ then both agents have exactly one item from each other $M_j$.
This implies that 

\begin{gather*}
v_i(A_w)\geq v_i(A_b\setminus \{\alpha\})-\sum_{h=1}^{l-1} v_i(A_b\cap M_{hr+1}) \\
v_i(A_w)\geq v_i(A_b\setminus \{\alpha\})-\sum_{h=1}^{l-1} \frac{2}{r-1}v_i(A_{w}^{h})
\end{gather*}

Observe that $\sum_{h=1}^{l-1} v_i(A_{w}^{h})=v_i(A_w)-v_i(A_w^l)\leq v_i(A_w)$. Then,

\begin{gather*}
v_i(A_w)\geq v_i(A_b\setminus \{\alpha\})-\frac{2}{r-1}v_i(A_w) \\
(1+\frac{2}{r-1})v_i(A_w)\geq v_i(A_b\setminus \{\alpha\}) \\
v_i(A_w)\geq (1-\frac{2}{r+1})v_i(A_b\setminus \{\alpha\})
\end{gather*}
This completes the proof that $A$ is $(1-\frac{2}{r+1})$-EF1.     
\end{proof}

\section{General Additive Valuations}\label{section: Additive}

  In this section we consider general additive valuations. We know already from the previous section that for exact EF1 allocations, we need $m\leq n\left\lfloor \frac{n}{\Delta} \right\rfloor +n-\Delta$. Here we restrict our study even further to the case of  $m\leq 2n$. To motivate this, recall that since $\Delta < n$,
  Proposition~\ref{counter-example} already implies that even when the agents have identical additive valuations, if $\Delta>\frac{2}{3}n$ and $m>2n-\Delta$, an EF1 allocation may not exist. We therefore find it important to understand EF1 existence in the still non-trivial case when $m$ is bounded by a small multiple of $n$.  

We first show that even for general additive valuations, an EF1 allocation can be computed when $m \leq 2n-\Delta$. Recall that this is the tight bound in the number of items for EF1 allocations established in the previous section.   
In particular, Algorithm~\ref{RR-Match}, which we analyze in the sequel, returns an EF1 allocation when $m\leq 2n-\Delta$ or when $\Delta\leq \frac{n}{2}$ and $m\leq 2n$. 

\begin{algorithm}[!ht]
\caption{}
\label{RR-Match}
\begin{algorithmic}[1] 

\Require $(N,M,V,G)$

\State Set $A_i=\emptyset$ for every agent $i$.
\State Run 1 round of Round-Robin so that each agent gets 1 item. \label{round-robin}
\State Let $M_2$ be the unallocated items. Construct $F(A,M_2)$.
\State Find a maximum size matching $R$ in $F(A,M_2)$. \label{RRline}
\State Allocate($A,R$)
\State \textbf{return} $A=(A_1,A_2,\ldots,A_n)$

\end{algorithmic}
\end{algorithm}


\begin{restatable}{theorem}{RRTHM}
\label{RR-Thm1}  For an instance $(N,M,V,G)$ with additive valuations, Algorithm~\ref{RR-Match} runs in polynomial time and returns an EF1 allocation if one of the following conditions holds
\begin{itemize}
    \item[(i)]  $m\leq 2n-\Delta$.
    \item[(ii)]  $\Delta\leq \frac{n}{2}$ and $m\leq 2n$.
\end{itemize}
\end{restatable}

\begin{proof}
    (i) Let $M_1$ be the items allocated in line \ref{round-robin}. In $F(A,M_2)$ a matching  saturating every vertex of $M_2$ always exists since $|M_2|\leq n-\Delta$ and every item of $M_2$ is feasible with at least $n-\Delta$ agents.
    Every bundle $A_i$ has either one item which belongs to $M_1$ or two items, one belonging to $M_1$ and one to $M_2$. 
    If an agent $i$ envies another agent $i^{\prime}$ and $|A_{i^{\prime}}|=2$ then $A_{i^{\prime}}=\{g_1,g_2\}$ where $g_1\in M_1$ and $g_2\in M_2$.
    Then $v_i(A_i)\geq v_i(A_{i^{\prime}}\setminus\{g_1\})$ otherwise agent $i$ would have chosen $g_2$ instead of her item with maximum value chosen in line 4 of Algorithm~\ref{RR-Match}.
    Thus, $A$ is EF1. Also, if $|A_{i^{\prime}}|=1$, then agent $i$ trivially satisfies the EF1 condition w.r.t. $i^{\prime}$.

    (ii) 
    By Proposition~\ref{Hall-Cor} since $\Delta\leq \frac{n}{2}$ a perfect matching in $F(A,M_2)$ always exists because every agent is feasible with at least $n-\Delta\geq n-\frac{n}{2}\geq \Delta$ and every item is feasible with at least $n-\Delta$ agents.
    The rest of the proof is in the same manner as in (i). 
    \end{proof}

If $m\geq 2n-\Delta+1$ and $\Delta>\frac{n}{2}$, Algorithm~\ref{RR-Match} cannot work. In particular, there is no guarantee that a matching saturating every vertex of $M_2$ in $F(A,M_2)$ exists in line~\ref{RRline} of Algorithm~\ref{RR-Match} (a property that is used for the cases handled by the proof of Theorem \ref{RR-Thm1}). To remedy this, we propose
Algorithm~\ref{RR-match-v2}, a technically more involved algorithm, as an extension of Algorithm~\ref{RR-Match}, that returns an EF1 allocation when $m\leq 2n$ and $\Delta\leq\frac{2n}{3}$. This results in the following theorem, which is the main technical result of this section. 

\begin{restatable}{theorem}{THMEFONE}
\label{Thm EF1}
    In an $(N,M,V,G)$ instance where $\Delta\leq\frac{2n}{3}$ and $m\leq 2n$, Algorithm~\ref{RR-match-v2} runs in polynomial time and returns an EF1 allocation.
\end{restatable}

\subsection{Prelims to Algorithm~\ref{RR-match-v2}: Relevant definitions and procedures}
Before presenting the pseudocode of Algorithm~\ref{RR-match-v2}, we provide some necessary concepts.

In a bipartite graph $G = (X\cup Y,E)$, if no matching exists that saturates all vertices of $X$, then Hall's condition of Theorem \ref{Hall-Thm} is violated by at least one subset $S$ and its neighborhood $N(S)$. A Hall violator of maximum deficiency is a Hall violator $(S,N(S))$, where $S\subseteq X$, that maximizes $|S|-|N(S)|$. 
A Hall violator of maximum deficiency $(S,N(S))$ is maximal if for every Hall violator $(S^{\prime}, N(S^{\prime}))$ where $S\subset S^{\prime}$ it holds that $|S|-|N(S)|>|S^{\prime}|-|N(S^{\prime})|$. 
This implies that if $(S,N(S))$ is a Hall violator of maximum deficiency that is also maximal, then $|X^{\prime}|< |N(X^{\prime})\cap (Y\setminus N(S))|$ for every $X^{\prime}\subseteq X\setminus S$, otherwise $S$ could be extended by $X^{\prime}$ and would not be maximal. 


We refer to a pair of items $p = \{g^1,g^2\}$ as feasible if $(g^1,g^2)$ is not an edge of the conflict graph $G$. Algorithm \ref{RR-match-v2} is based on constructing the following type of matchings.  

\begin{definition}
A feasible 2-matching $R$ in $F(A,U)$ w.r.t. $S$ where $S\subseteq U$, is a set of edges of $F(A,U)$ such that 
\begin{itemize}
    \item[(i)] each $g\in U$ is matched with exactly one agent.
    \item[(ii)] each agent $i$ is matched with at most 2 items of $U$.
    \item[(iii)] if $(i,g), (i,g^{\prime})\in R$ then $g,g^{\prime}\in S$ and $(g,g^{\prime})\notin E(G)$.
\end{itemize}
\end{definition}

The set of feasible pairs of a 2-matching $R$ is the set of pairs of items that matched to the same agent in $R$.

Another graph that we make use of is the \textit{envy graph} of an allocation $A$. This is a directed graph, denoted by $EG(A)$, where its vertex set is the set of agents and there is an arc $(i,j)$ from agent $i$ to $j$ whenever $v_i(A_j)>v_i(A_i)$.

Within our algorithm, we make use of Procedure~\ref{envy-swap} stated below, which receives as input a partial allocation $A$ and a bundle $A_f$ that does not belong to any agent. 
The agents have the option to swap their bundle with $A_f$, and when the procedure ends no agent envies the final remaining bundle, $A_f$. We also ensure that $EG(A)$ is acyclic at the end (due to the order of the swaps).

\begin{procedure}[tbh]
\caption{Envy Swap($A,A_f$)}
\label{envy-swap}
\begin{algorithmic}[1]

\Require $A = (A_1,\dots, A_n)$, $A_f$

\State Transform $EG(A)$ into DAG by Envy Cycle Elimination \cite{lipton2004approximately}.
\State Let $1,2,\ldots,n$ be a topological order of $EG(A)$.
\For{$i=n$ \textbf{downto} $1$}
    \If{$v_i(A_i) < v_i(A_f)$}
        \State Swap $A_i$ with $A_f$.
    \EndIf
\EndFor
\State \textbf{return} $A$, $A_f$

\end{algorithmic}
\end{procedure}



When each agent has 1 item in her bundle and a matching does not exist, the question is how can we handle the items of the Hall violator $S$.
Algorithm~\ref{RR-match-v2} creates pairs of items of $S$ in a feasible manner and each agent has the choice to exchange her bundle with one of these pairs.
This is repeated until either the unallocated items can be matched to the agents with a single item so far or no agent is willing to exchange her bundle with such a pair.
In the latter case some agents will end up with 3 items in their bundle and others with 1.

\begin{algorithm}[!ht]
\caption{}
\label{RR-match-v2}
\begin{algorithmic}[1] 

\Require $(N,M,V,G)$

\State Set $A_i=\emptyset$ for every agent $i$.
\State Run 1 round of Round-Robin so that each agent gets 1 item.
\State Let $M_2$ be the set of unallocated items. Construct $F(A,M_2)$. \label{line:construction}

\If{a maximum size matching $R$ in $F(A,M_2)$ saturating every vertex in $M_2$ exists} \label{line:matching}
    \State Allocate($A,R$)
\Else
    \State Set $A^{\prime}=A$, $N^{\prime}=N$. 
    \State Find a maximal Hall violator of maximum deficiency $(S,N(S))$ where $S\subseteq M_2$ in $F(A^{\prime},M_2)$. \label{line: Hall-violator}
    \State $c = \text{true}$
    
    \While{$c = \text{true}$} \label{while3-start}
        \State Compute a feasible 2-matching $R^{\prime}$ w.r.t. $S$ in $F(A^{\prime},M_2)$. 
        \State Let $P$ be the set of feasible pairs of $R^{\prime}$. \label{line:2-matching}
        
        \If{$v_h(p)>v_h(A_h)$ for some feasible pair $p\in P$ and some agent $h\in N$}. \label{line: swap h}
            \State Set $A_f=A_h$. \State $A_h=\{p_h\}$ where $p_h\in P$ and $p_h=\arg\max_{p\in P} v_h(p)$. \label{line: exchange}
            \State Envy Swap$(A,A_f)$. The new bundle $A_f$   remains unallocated. \label{line: envy-swap}
            \State Let $N^{\prime}=\{i\in N\mid |A_i|=1\}$ and $A^{\prime}$ the respective partial allocation for the agents of $N^{\prime}$. 
            \State Let $M_2$ be the set of unallocated items. Construct $F(A^{\prime},M_2)$. \label{line: reconstruct}
            
            \If{a maximum size matching $R$ in $F(A^{\prime},M_2)$ saturating every vertex in $M_2$ exists} \label{line: match}
                \State Allocate($A,R$)
                \State $c = \text{false}$ \label{line:21}
            \Else 
                \State Compute a maximal Hall violator of maximum deficiency $(S,N(S))$ in $F(A^{\prime},M_2)$. \label{line: decide S} 
            \EndIf
        \Else 
            \State Allocate($A,R^{\prime})$ \label{line:allocate2}
            \State $c = \text{false}$
        \EndIf
    \EndWhile \label{while3-end}
\EndIf

\State \textbf{return} $A$

\end{algorithmic}
\end{algorithm}

   \subsection{Auxiliary graph and computing a feasible 2-matching}
    \label{app-subsec:matching}

    Before proving that Algorithm~\ref{RR-match-v2} returns an EF1 allocation, we describe how to identify a maximal Hall violator of maximum deficiency and we provide a procedure for the construction of the feasible $2$-matching in line~\ref{line:2-matching}. Both rely on the notion of an auxiliary graph, as defined below.
\begin{definition}
\label{def:auxiliary}
Given a set of edges $R$ in $F(A^{\prime},M_2)$, the auxiliary graph $F_R(A^{\prime},M_2)$ is the directed graph obtained by orienting the edges of $F(A^{\prime},M_2)$  such that every edge $e = \{i, g\}$ with $i \in N'$ and $g \in M_2$  is directed from $i$ to $g$ if $e \in R$, and from $g$ to $i$ if $e \notin R$.    
\end{definition}    
Note that the unmatched vertices in $R$ are single-vertices strongly connected components in the auxiliary graph, since either they may only have ingoing edges (if they belong to $N^{\prime}$) or only outgoing (if they belong to $M_2$).
These strongly connected components are regarded as trivial.
The non-trivial strongly connected components are those that are induced by the matched vertices.

The structural existence of the maximal Hall violator of maximum deficiency follows from the classic Dulmage-Mendelsohn decomposition \cite{Dulmage_Mendelsohn_1958}, and it can be efficiently computed using standard alternating path reachability \cite{lovasz1986matching}.
For the sake of completeness, we provide a description of how the maximal Hall violator of maximum deficiency $(S,N(S))$ is decided. Firstly we construct $F_R(A^{\prime},M_2)$ given a maximum size matching $R$ in $F(A^{\prime},M_2)$.
Let $T$ be the set that contains every $g\in M_2$ such that there exists an $i\in N^{\prime}$ which is unmatched in $R$ and $i$ is reachable from $g$ in $F_R(A^{\prime},M_2)$.
A maximal Hall violator of maximum deficiency $(S,N(S))$ is the one where $S=M_2\setminus T$ .

The DAG of the strongly connected components of a directed graph is formed by contracting each strongly connected component into a single vertex, preserving only the edges that connect different components.

Now we present Procedure~\ref{procedure 2-matching} that decides a feasible 2-matching $R^{\prime}$ w.r.t. $S$ in $F(A^{\prime},M_2)$. 
As part of the proof of  Theorem~\ref{Thm EF1} that follows in Section \ref{subsec:proofsketch}, it is shown in Lemma~\ref{claim1} that this procedure yields a feasible $2$-matching at each iteration of line~\ref{line: Hall-violator} in Algorithm~\ref{RR-match-v2}.

\begin{procedure}
\caption{Computing a feasible 2-matching $R^{\prime}$ w.r.t. $S$ in $F(A^{\prime},M_2)$.}
\label{procedure 2-matching}
\begin{algorithmic}[1]

\Require $F(A^{\prime},M_2)$, a maximum size matching $R$, a maximal Hall violator of maximum deficiency $(S,N(S))$. 

\State Construct the auxiliary graph $F_R(A^{\prime},M_2)$, according to Definition \ref{def:auxiliary}.
\State Compute its strongly connected components. \label{prod-auxiliary}
\State Let $B_s$ be the subgraph induced by the vertices of $S$ and $N(S)$ belonging in a non-trivial strongly connected component which is a sink in the DAG of strongly connected components of $F_R(A^{\prime},M_2)$. \label{prod-sink}
\State Let $R_s$ be the subset of $R$ belonging to $B_s$. \label{prod-edges in sink} 
\State Let $S_{un}$ be the vertices of $S$ that are unmatched in $R$. \label{prod-unmatched}
\State For every agent $i\in N^{\prime}$, let $g_i$ be the single item belonging to her bundle $A_i$.
\State Construct a bipartite graph $G_2=(S_{un}\cup R_s, E_2)$ that has one vertex $g^{\prime}$ for each $g^{\prime}\in S_{un}$, one vertex $(i,g)$ for each $(i,g)\in R_s$ and $g^{\prime}$ is connected with $(i,g)$ if and only if $g^{\prime}$ is neighbor to neither $g$ nor $g_i$ in the conflict graph $G$ i.e. $(g,g^{\prime}),(g^{\prime},g_i)\notin E(G)$. \label{prod-construct G2}
\State Find a matching $R_2$ in $G_2$ saturating $S_{un}$. 
\State Set $R^{\prime}=R$.
\For{each $g^{\prime}$ matched with $(i,g)$ in $R_2$}
    \State $R^{\prime}=R^{\prime}\cup (i,g^{\prime})$
\EndFor

\State \textbf{return} $R^{\prime}$

\end{algorithmic}
\end{procedure}

\subsection{Proof sketch of Theorem~\ref{Thm EF1}}
\label{subsec:proofsketch}

    We provide here the main skeleton of our proof. All missing proofs of subsequent lemmas and claims that we state and use here can be found in the Appendix. 
    
    Because of Theorem~\ref{RR-Thm1}, we assume that $\Delta>\frac{n}{2}$ and $m>2n-\Delta$.
    If a matching $R$ in line~\ref{line:matching} exists, Algorithm~\ref{RR-match-v2} does not enter the while loop in lines~\ref{while3-start}-\ref{while3-end}. The returned allocation $A$ is EF1. 
    The rest of this part of the proof is as in Theorem~\ref{RR-Thm1}.
    Otherwise, there is a Hall violator $(S,N(S))$. 
    
    Since the maximum degree of $G$ is $\Delta$, every vertex of $F(A,M_2)$ has degree at least $n-\Delta$.
    For the Hall violator when Algorithm~\ref{RR-match-v2} enters the while loop in lines~\ref{while3-start}-\ref{while3-end} for the first time, it holds that $n-\Delta\leq |N(S)|<|S|\leq \Delta$.

    In every iteration of the while loop in lines~\ref{while3-start}-\ref{while3-end}, we consider the following four cases for the bundle $A_f$ that is left unallocated by Envy Swap in line~\ref{line: envy-swap}. 
    We denote by $M_s$ the set of items belonging to the agents of $N(S)$ and by $M_0$ the items belonging to the agents of $N^{\prime}\setminus N(S)$.
    These cases are distinguished based on which agent held the bundle before Envy Swap. The cases are the following.
    \begin{itemize}
        \item[C1:] $A_f=\{g\}$, where $g\in M_s$ and $g$ is not feasible with any item of $M_0$.
        \item[C2:]  $A_f=\{g\}$, where $g \in M_s$ and $g$ is feasible with some item of $M_0$.
        \item[C3:] $A_f=\{g\}$, where $g\in M_0$.
        \item[C4:] $A_f=\{p_i\}$, where $p_i$ is a feasible pair belonging to an agent $i\in N\setminus N^{\prime}$.
     \end{itemize}

    Observe that in the first three cases the number of agents that are allocated a feasible pair is increased by $1$ after the iteration.
    Moreover, in the first iteration $A_f$ cannot fall into Case C4 since no agent is allocated a feasible pair so far.
    
    To proceed with the proof, we first show that Algorithm~\ref{RR-match-v2} is well-defined and always terminates. 
    This is established by Lemma \ref{claim1} and Claims \ref{claimcombined}-\ref{claim4}. 
    Lemma~\ref{claim1} establishes the existence of the 2-matching in each iteration.

    \begin{restatable}{lemma}{firstclaim}
    \label{claim1}
       A feasible 2-matching $R^{\prime}$ w.r.t. $S$ in line~\ref{line:2-matching} in each iteration of the while loop in lines~\ref{while3-start}-\ref{while3-end} exists and can be computed in polynomial time.
    \end{restatable}



Claim~\ref{claimcombined} esteblishes some structural properties of the auxiliary graph $F(A^{\prime},M_2)$. Recall that $B_s$ is the strongly connected component which is non-trivial sink in the DAG formed by the strongly connected components of $F_R(A^{\prime},M_2)$.

\begin{restatable}{claim}{claimcombined}
    
\label{claimcombined}
The following hold:
\begin{enumerate}
    \item[(i)]  Every $g\in M_0$ is feasible with at least one other item of $M_0$.
    \item[(ii)]  Every $g\in S$ is feasible with at least one other item of $M_s\cap B_s$.
\end{enumerate}
\end{restatable}

The specific construction of the 2-matching, combined with the structure of the feasibility graph implied by the condition that $\Delta \leq \frac{2n}{3}$, ensures the validity of the claims stated below. 
Claim~\ref{claim2} specifies the conditions under which Algorithm~\ref{RR-match-v2} exits the while loop in lines~\ref{while3-start}-\ref{while3-end} and terminates.

\begin{restatable}{claim}{fourthclaim}
\label{claim2}
    Algorithm~\ref{RR-match-v2} exits the while loop in lines~\ref{while3-start}-\ref{while3-end} if one of the following holds.
    \begin{itemize}
        \item[(i)] $v_h(p)\leq v_h(A_h)$ for every feasible pair $p\in P$ and any agent $h\in N$ in line~\ref{line: swap h}. 
        \item[(ii)] $|S|-|N(S)|=1$ and the bundle $A_f$ returned by Envy Swap falls in Case C2.
        \item[(iii)] $|S|-|N(S)|=1$ and the bundle $A_f$ returned by Envy Swap falls in Case C3.
         \item[(iv)] $|S|-|N(S)|=2$ and $A_f$ returned by Envy Swap falls in Case C3.
    \end{itemize}
\end{restatable}

Claim~\ref{claim3} shows that the items belonging to the agents of $N^{\prime}\setminus N(S)$ is a subset of the respective items in the initial iteration.

    \begin{restatable}{claim}{fifthclaim}
    \label{claim3}
        When Algorithm~\ref{RR-match-v2} re-enters the while loop in lines~\ref{while3-start}-\ref{while3-end}, the set of items belonging to the agents of $N^{\prime}\setminus N(S)$ is a subset of the respective set in the previous iteration.
   \end{restatable}

 The number of agents that have been already allocated a feasible pair is denoted by $n_p$. i.e.,  $n_p=|N\setminus N^{\prime}|$. Clearly the number of items that have been allocated to the agents of  $N\setminus N^{\prime}$ is $2n_p$.
Claim~\ref{claim4} is obtained almost directly by Claim~\ref{claim3}.

    \begin{restatable}{claim}{sixthclaim}
    \label{claim4}
    If $|S|=\Delta-2n_p$ the bundle $A_f$ returned by Envy Swap cannot fall into Case C1.
    \end{restatable}

The set of items belonging to the agents of $N^{\prime}\setminus N(S)$ are connected in the conflict graph with every item belonging to $S$ at some point during the execution of Algorithm~\ref{RR-match-v2}. Thus, Claim~\ref{claim3} implies that the overall number of items which are elements of $S$ at some iteration is at most $\Delta$.   
Since $\Delta\leq \frac{2n}{3}$, the number of feasible pairs that can be constructed is $O(n^2)$.
    Moreover, observe that the value of any agent does not decrease during Algorithm~\ref{RR-match-v2} and in every iteration of the while loop at least one agent strictly increases her value.
    Consequently, in combination with Claim~\ref{claim4}, after a polynomial number of iterations the deficiency $|S|-|N(S)|$ would strictly decrease if Algorithm~\ref{RR-match-v2} re-enters the while loop.
    Thus, because of Claim~\ref{claim2}, Algorithm~\ref{RR-match-v2} always terminates.

Having established the termination of the algorithm, we can now establish that it produces a desirable allocation. 
    \begin{restatable}{lemma}{lemmaEFONE}\label{lemEF1}
        Algorithm~\ref{RR-match-v2} returns an EF1 allocation, under the assumptions stated in Theorem \ref{Thm EF1}.
    \end{restatable}

    \begin{proof}
     We now show why the allocation is EF1.
    The Round Robin manner of how each agent is allocated the first of item of her bundle  alongside with the Envy Swap guarantee that after every iteration of while loop each agent prefers his item more than any item of $M_2$ and every single item of agents that have already allocated two items i.e. the agents of $N\setminus N^{\prime}.$
    Thus, the partial allocation after each time an agent $h$ exchanges her bundle with a feasible pair is EF1 since, even if an agent $i$ envies agent $i^{\prime}\in N\setminus N^{\prime}$ who have two items then $v_i(A_i)\geq v_i(A_{i^{\prime}}\setminus \{g\})$ for each $g\in A_{i^{\prime}}$. Otherwise, $i$ would have been allocated $g$ either in Round Robin or in Envy Swap.
    The agents of $N\setminus N^{\prime}$ when Algorithm~\ref{RR-match-v2} terminates are not allocated other item, thus the EF1 condition holds for agents who may envy them.
    The bundle of an agent $i^{\prime}\in N^{\prime}$ who is allocated either an item $g$ or a feasible pair $p$ when Algorithm~\ref{RR-match-v2} terminates is $A_{i^{\prime}}=\{g_i, g\}$ or $A_{i^{\prime}}=\{g_i, p\}$.
    For an agent $i$ who may envy $i^{\prime}$ it holds that $v_i(A_i)\geq v_i(A_{i^{\prime}}\setminus g_i)$. Otherwise $i$ would have been allocated $g$ either in Round Robin or in Envy Swap in the case where $A_{i^{\prime}}=\{g_i, g\}$ or $i$ would have satisfied the condition in line~\ref{line: envy-swap} and would have exchanged his bundle with $p$  if $A_{i^{\prime}}=\{g_i, p\}$.
    This concludes that allocation $A$ returned by Algorithm~\ref{RR-match-v2} is EF1. 
\end{proof}

  Finally, we argue that Algorithm~\ref{RR-match-v2} runs in polynomial time. 
    A maximal Hall violator of maximum deficiency and a feasible 2-matching can be computed in polynomial time (see Appendix \ref{app-subsec:matching}).
    Since the initial deficiency $|S|-|N(S)|$ is at most $\Delta-(n-\Delta)\leq \frac{n}{3}$ and after a polynomial number of iterations the deficiency is strictly decreased, then the overall number of iterations  that the while loop is executed is polynomial. This concludes the proof. 
    

\subsection{Further consequences of Algorithm~\ref{RR-match-v2}}
The following theorem shows that when $\Delta> \frac{2n}{3}$ and $m=2n-\Delta+1$, a necessary condition for the non-existence of EF1 allocations, even with additive valuations, is that $K_{\Delta,2n-2\Delta+1}$ is subgraph of the conflict graph $G$. Recall that $K_{\Delta,2n-2\Delta+1}$ is the conflict graph of the counterexample in Proposition~\ref{counter-example} for every $\Delta>\frac{2n}{3}$. The proof is similar to the one of Theorem~\ref{Thm EF1} in a simplified manner and is provided in the Appendix.  

\begin{restatable}{theorem}{THMBIGDELTA} \label{thm:4.3}
   In a $(N,M,V,G)$ instance when $m=2n-\Delta+1$, $\Delta\geq\frac{2n+1}{3}$ and $G$ does not have $K_{\Delta,2n-2\Delta+1}$ as a subgraph, Algorithm~\ref{RR-match-v2} returns an EF1 allocation.
\end{restatable}

\begin{proof}
    The proof is similar to proof of Theorem~\ref{Thm EF1} in a simplified manner. We follow the same notation as in the proof of Theorem~\ref{Thm EF1}.
    In line~\ref{line:matching} of Algorithm~\ref{RR-match-v2} since $m=2n-\Delta+1$ if there is no matching for the Hall violator it holds that $|S|=n-\Delta+1$, $|N(S)|=n-\Delta$ and $|N\setminus N(S)|=\Delta$ since every item is feasible with at least $n-\Delta$ agents. Since $|S|-|N(S)|=1$ and because of Claim~\ref{claim2}, Algorithm~\ref{RR-match-v2} will re-enter the while loop if $A_f$ falls into Case C1.

    Since the deficiency is never increased during the execution of Algorithm~\ref{RR-match-v2}, while no matching saturating every vertex of $M_2$ exists it holds that $|S|=n-\Delta-n_p+1$, $|N(S)|=n-\Delta-n_p$ and $|N^{\prime}\setminus N(S)|=\Delta$. Because of Claim~\ref{claim2}, Algorithm~\ref{RR-Match} will exit the while loop and will terminate in the first time that $A_f$ falls into Cases C2 or C3. 
    It continues while it falls into Cases C1 or C4. Thus, and in combination with Claim~\ref{claim3}, in every iteration $N(S)$ is a subset of the respective set in the previous iteration either with an item less (Case C1) or exactly the same (Case C4).
    
     Since in every iteration at least one agent strictly increases her value, after a polynomial number of iterations, Case C4 would be rendered impossible, provided that no item that had never  belonged to $S$ in the previous iteration appears in the set for the first time.
     
   The fact that $G$ has not $K_{\Delta,2n-2\Delta+1}$ as subgraph implies that from the $n-\Delta$ items belonging to the agents of $N(S)$ in the first iteration, at most $n-\Delta-1$ of them may be connected with every item belonging to the agents of $N\setminus N(S)$.
    This implies that Case C1 is can appear at most $n-\Delta-1$ times during the execution of Algorithm~\ref{RR-match-v2}.
    Thus, if $n_p=n-\Delta-1$ observe that the picture of the $F(A^{\prime}, M_2)$ would be the following:   $|S|=2$, $N(S)=1$ and $|N^{\prime}\setminus N(S)|=\Delta$ where there is at least one item belonging  to the agents of $N^{\prime}\setminus N(S)$ that is feasible with the one belonging to the single agent of $N(S)$ and Case C1 for $A_f$ is not possible. Thus, Algorithm~\ref{RR-match-v2} will terminate either after one iteration if $A_f$ falls into either Case C2 or Case C3 or after polynomial number of iterations if in all of them $A_f$ falls into Case C4.

    The proof that the returned allocation $A$ is EF1 and that Algorithm~\ref{RR-match-v2} runs in polynomial time is similar to the Theorem~\ref{Thm EF1}. 
\end{proof}

\section{Conclusions}

We have studied the existence of complete EF1 allocations in constrained fair division with a conflict graph. Our main technical contribution is the use of matching theory (and especially Hall's theorem) 
for establishing the existence and efficient computation of EF1 allocations for certain classes of valuation functions. The provided results demonstrate that the maximum degree of the conflict graph plays a crucial role in determining the existence of such fair allocations. An open question arising from our work is for which $m$ an EF1 allocation is guaranteed in the case that $\Delta\leq 2n/3$ and $m>2n$. Moreover, developing more refined approximation algorithms is also an interesting direction for future work. 

Beyond the open questions that follow directly from our results, there are other important research avenues as well. The maximum degree $\Delta$ provides only partial information about the conflict graph. Focusing on this parameter allows us to establish a robust worst-case analysis. However, a graph with a small number of vertices of degree $\Delta$ may differ significantly in the conflicts it generates. Investigating how other structural properties, such as the average degree or vertex connectivity could affect and strengthen the fairness guarantees (maybe in combination with $\Delta$) is an appealing direction for future work. Furthermore, our work has focused only on the EF1 criterion, which is already quite challenging, but several other fairness notions can also be explored. There has been some progress on this front (e.g., for the MMS criterion in \cite{Hummel2021}), but we are still far from having a more complete understanding on which notions are more amenable under this constrained model.

\newpage
{\Large \textbf{Appendix}}

\setcounter{section}{0}

\renewcommand{\thesection}{\Alph{section}}
\renewcommand{\thetheorem}{\thesection.\arabic{theorem}}
\renewcommand{\thelemma}{\thesection.\arabic{lemma}}
\renewcommand{\theproposition}{\thesection.\arabic{proposition}}
\renewcommand{\theequation}{\thesection.\arabic{equation}}
\renewcommand{\thefigure}{\thesection.\arabic{figure}}

\setcounter{claim}{0} 
\renewcommand{\theclaim}{\Alph{claim}}

\makeatletter
\@addtoreset{theorem}{section}
\@addtoreset{lemma}{section}
\@addtoreset{proposition}{section}
\@addtoreset{equation}{section}
\makeatother

\section{Missing Parts from Section~\ref{section: Preliminaries}}
\subsection{Proof of Proposition \ref{Hall-Cor}}

\HallCor*

\begin{proof}
   Assume that in $G$ there is no matching saturating every vertex of $X$.
   By Hall's Theorem there exist some $S\subseteq X$ where $|N(S)|<|S|$. 
   If $|S|\leq n^{\prime}-d$, since every vertex in $X$ has degree at least $n^{\prime}-d$ then $|N(S)|\geq n^{\prime}-d\geq |S|.$
   If $|S|>n^{\prime}-d$, then $|X\setminus S|<d$.
   Let $y$ be a vertex belonging to $Y\setminus N(S)$. 
   Since $y$ has degree at least $d$ and  $|X\setminus S|<d$, then by the Pigeonhole Principle it is not possible for $y$ to be connected only to vertices of $X\setminus S$ and this leads to a contradiction and concludes that $G$ has a matching saturating every vertex of $X$.  
\end{proof}



    \section{Missing Parts from Section~\ref{section: Additive}}

\subsection{Missing parts from Proof of Theorem~\ref{Thm EF1}}

  When Algorithm \ref{RR-match-v2} entered  the while loop for the first time it holds that $n_p=0$. 
Recall that $n_p=|N\setminus N^{\prime}|$ i.e. the number of agents that have already been allocated a feasible pair and they are not in $F(A^{\prime},M_2)$.
We prove Lemma~\ref{claim1} and  Claims~\ref{claim2}, \ref{claim3}, and \ref{claim4} by induction in the iterations of the while loop in lines~\ref{while3-start}-\ref{while3-end} of Algorithm~\ref{RR-match-v2}. 
Additionally, we state and prove Claim~\ref{claimcombined} which establishes some structural properties of the feasibility graph in each iteration that are necessary in the proofs of the other claims.
We show them in a unified manner for the first and the next iterations. 
We assume that all of them hold after several iterations and we prove that they also hold in the next one.
For the first iteration, which serves as the basis of the induction, we have $n_p=0$ and the only difference in the proof of them is that $A_f$ returned by Envy Swap cannot fall into Case C4 in the first iteration since no agent is allocated a feasible pair so far.

Before proving Lemma~\ref{claim1} and  Claims~\ref{claimcombined}, \ref{claim2}, \ref{claim3}, and \ref{claim4} , we provide a description of $F(A^{\prime},M_2)$ after several iterations of the while loop.
In $F(A^{\prime},M_2)$ each part has $n-n_p$ vertices.
Since the maximum degree in the conflict graph is $\Delta$, every vertex in $F(A^{\prime},M_2)$ has degree at least $n-\Delta-n_p$.
If there is no matching in $F(A^{\prime},M_2)$, for the Hall violator it holds $n-\Delta-n_p\leq |N(S)|<|S|$.
Moreover, because Claim~\ref{claim2} holds in the previous iterations, the items of $M_0$ are connected in the conflict graph with the items of $M_p$. 
Thus, $|S|\leq \Delta-2n_p$ and $|N(S)|<\Delta-2n_p$ and $|N^{\prime}\setminus N(S)|>n-n_p-(\Delta-2n_p)=n-\Delta+n_p.$
Recall that $M_s$ is the items belonging to the agents of $N(S)$ and $M_0$ those to the agents of $N^{\prime}\setminus N(S)$.
Every item of $S$ is feasible with at least $n-\Delta-n_p$ items of $M_s$ and each item of $M_0$ is feasible with at least $n-\Delta+n_p$ of $M_2\setminus S$ since it has $2n_p$ conflicts with the items of $M_p$.

Since in $F(A^{\prime},M_2)$ each part has $n-n_p$ vertices and $n-\Delta-n_p\leq |N(S)|<|S|\leq \Delta-2n_p$, let $|S|=n-\Delta-n_p+r+k$ and $|N(S)|=n-\Delta-n_p+r$ for some $r\geq 0$ and $k\geq 1$.

In the proof of the  Lemma~\ref{claim1} and  Claims~\ref{claimcombined}, \ref{claim2}, \ref{claim3}, and \ref{claim4}, because in the Envy Swap the agents that will remain in $F(A^{\prime},M_2)$ may have swap their bundles in the Envy Swap, we may consider that in $F(A^{\prime},M_2)$ the vertices of the one part is not the agents $i\in N^{\prime}$ but their respective items $M_s\cup M_0$. Equivalently to this notation,  if $A_i=\{g_i\}$ and $(i,g^{\prime})\in R^{\prime}$ for some agent $i\in N^{\prime}$, we consider that $(g_i,g)\in R^{\prime}$.

\firstclaim*
\begin{proof}
The fact that $|S|\leq \Delta-2n_p$ implies 
\begin{gather*}
n-\Delta-n_p+r+k\leq \Delta-2n_p\\
    k\leq 2\Delta-n-n_p-r
\end{gather*}
Since $(S, N(S))$ is a maximal Hall violator of maximum deficiency, the number of unmatched items $S_{un}$ of $S$ in a maximum size matching $R$ in $F(A^{\prime},M_2)$ is $|S|-|N(S)|=|S_{un}|=k$.
In the auxiliary graph constructed in line~\ref{prod-auxiliary}, the non-trivial sink $B_s$ in the DAG of strongly connected components consists of at least $n-\Delta-n_p$ edges of $R$ i.e. $|R_s|\geq n-\Delta-n_p$. 
This holds because each item in $S$ is feasible with at least $n-\Delta-n_p$ agents of $N^{\prime}$, thus the items of $S$ belonging to $B_s$ are feasible with at least $n-\Delta-n_p$ belonging to $B_s$ too.

Because $N^{\prime}=n-n_p$ and $|N(S)|=n-\Delta-n_p+r$, it holds that $|N^{\prime}\setminus N(S)|=\Delta -r$.
Since, every item of $S$ is connected in the conflict graph $G$ with every item of $M_0$, it follows that every item in $S$ may have at most $r$ other items with which is connected in $G$.

In $G_2$ constructed in line~\ref{prod-construct G2} it holds that $|S_{un}|=k$ and $|R_s|\geq n-\Delta-n_p$.
Moreover, since every items of $S$, so also of $S_{un}$ may have at most $r$ other other items with which is connected in $G$ except those belonging to the agents of of $N^{\prime}\setminus N(S)$, it follows that every item of $S_{un}$ is feasible with at least $n-\Delta-n_p-r$ vertices $(i,g)$ of $R_s$ in $G_2$.
To prove the existence of a matching $R_2$ saturating $S_{un}$ in $G_2$, and thus the existence of a feasible 2-matching w.r.t $S$ in $F(A^{\prime},M_2)$, it suffices to show that $k\leq n-\Delta-n_p-r$.
Since, $k\leq 2\Delta-n-n_p-r$ it suffices 
\begin{gather*}
    2\Delta-n-n_p-r\leq  n-\Delta-n_p-r\\
    2n-3\Delta\geq 0    
\end{gather*}
The fact that $\Delta\leq \frac{2n}{3}$  concludes the proof of the Lemma. 

\end{proof}

\claimcombined*

\begin{proof}
(i) Since $|M_0|=|N^{\prime}\setminus N(S)|>n-\Delta-n_p$, it holds that $|M_0|\geq n-\Delta-n_p+1$.
An item  $g$ of $M_0$ is in conflict with the items of $S$ and the items of $M_p$. This implies that it may have at most $C=\Delta-|S|-|M_p|$ other conflicts.
Since $|S|=n-\Delta-n_p+r+k$ and $|M_p|=2n_p$ it follows that 
\begin{align*}
   C= \Delta-|S|-|M_p| &=\Delta-(n-\Delta-n_p+r+k)-2n_p\\
    &=2\Delta-n-n_p-r-k
\end{align*}
To conclude the proof it suffices to show that $|M_0\setminus g|>2\Delta-n-n_p-r-k$ because then $g$ cannot be in conflict with every other item of $M_0$. 
The fact that $|M_0\setminus g|\geq n-\Delta+n_p$ implies that 
\begin{align*}
    |M_0\setminus g|-C&=(n-\Delta+n_p)-(2\Delta-n-n_p-r-k)\\
    &=(2n-3\Delta)+r+k+2n_p.
\end{align*}
The fact that $\Delta\leq \frac{2n}{3}$, $r,n_p\geq 0$ and $k\geq 1$ implies that $|M_0\setminus g|-(\Delta-|S|-|M_p)>0$ and this concludes the proof.

(ii) As described in the proof of Lemma~\ref{claim1}, $B_s$ contains at least $n-\Delta-n_p$ vertices of $N(S)$.
If $g\in S$ is not connected with any of them and since $g$ is connected with at least $n-\Delta-n_p$ vertices of $N(S)$, then $|N(S)|\geq 2(n-\Delta-n_p)$.
Since, $|N(S)|<\Delta-2n_p$ it follows that
\begin{align*}
    2(n-\Delta-n_p)&< \Delta -2n_p\\
    2n-2\Delta&<\Delta\\
    2n-3\Delta&<0
\end{align*}
which is a contradiction since $\Delta\leq\frac{2n}{3}$. 
\end{proof}

The proofs of Claims~\ref{claim2} and~\ref{claim3} rely on the construction of the $2$-matching $R^{\prime}$ and we shall utilize the auxiliary graph $F_{R^{\prime}}(A^{\prime}, M_2)$ obtained by orienting the edges of $F(A^{\prime}, M_2)$ as previously described or by reversing the direction of the edges of the matching $R_2$ of Procedure~\ref{procedure 2-matching} in $F_{R}(A^{\prime}, M_2)$.

\fourthclaim*
\begin{proof}
We show why for each of (i), (ii), (iii) and (iv)  Algorithm~\ref{RR-match-v2} exits the while loop.

(i). If $v_h(p)\leq v_h(A_h)$ for every feasible pair $p\in P$ and every agent $h\in N$ in line~\ref{line: swap h}, it is clear that Algorithm~\ref{RR-match-v2} allocates the remaining items via the $2$-matching $R^{\prime}$ in line~\ref{line:allocate2} and then exits the while loop.

In cases (ii), (iii) and (iv) we prove that after the reconstruction of the feasibility graph in line~\ref{line: reconstruct}, a matching saturating every vertex of $M_2$ exists in line~\ref{line: match}.
Let in these cases that agent $h$ left his bundle to receive a feasible pair $p$ and $A_f=\{g\}$ is the result of Envy Swap.
Let $x$ be the item or the pair matched with $g$ in the 2-matching $R^{\prime}$ and $g_p$ be the item $N(S)$ matched with $p$.
In $F(A^{\prime},M_2)$ as reconstructed in line~\ref{line: reconstruct} the items of $p$ are removed from $M_2$ while $g$ removed from $M_s\cup M_0$ and is added to $M_2$.
Let also $R_1$ be the subset of $ R^{\prime}$ after removing the edges $(g_i,g^1),(g_i,g^2)$ of $R^{\prime}$ where the vertices $g^1, g^2$ consist of a feasible pair $p^{\prime}$ for every feasible pair $p^{\prime}$.

(ii). In this case $g\in N(S)$. Let $u$ be the one unmatched vertex of $M_0$ and since $|S|-|N(S)|=1$ by the construction of the $(S,N(S))$ then $u$ is reachable by any item of $M_2\setminus S$.
To prove that a  matching saturating every vertex of $M_2$ exists in the reconstructed $F(A^{\prime},M_2)$ in line~\ref{line: match} it suffices to augment $R_1$ by an alternating path $P_1$ which starts from $g$ and ends to $u$ and if $x\neq p$ (then $x$ is a single item since $p$ is the only pair as $|S|-|N(S)|=1$) also by an alternating path $P_2$ where starts from $x$  and ends to $g_p$.
Since $A_f$ falls into Case C2, there is an item of $M_0$ which $q$ is feasible with $g$. If $q=u$ this concludes the existence of $P_2$, otherwise $q$ is matched in $R_1$ with some $q^{\prime}\in M_2\setminus S$.
Since $u$ is reachable by $q^{\prime}$ a path $P_2$ always exists.
Path $P_1$ is needed only if $x\neq p$, otherwise, if $x=p$, observe that $R_1$ augmenting by $P_1$ is a matching saturating every vertex $M_2$.
By the construction of $R^{\prime}$ it holds that $g_p$ is a vertex of $M_s\cap B_s$.
By Claim~\ref{claimcombined}, independently whether $x$ belongs to $B_s$ or not,  $x$ is feasible with a vertex $z$ of $M_s\cap B_s$ and thus there exists the path $P_2$ from $x$ to $g_p$ since $z$ and $g_p$ are both in the strongly connected component $B_s$ and only the  outgoing edges edges $g_p$ towards the vertices $p$ and also the ingoing edge of $x$ from $g$ have been removed.

(iii). This case is similar to case (ii). The only difference is that since $x\notin S$, the existence of the vertex $q$ for the path $P_1$ from $x$ to $u$ is guaranteed by Claim~\ref{claimcombined}.

(iv). Let $z_1,z_2$ the two items of $S$ that consist of the feasible pair $p_0$ which is the other feasible pair except $p$ since $|S|-|N(S)|=2.$ Let that $p_0$ is matched with $g_0$, i.e. $(z_1,g_0),(z_2,g_0)\in R^{\prime}$ and $(z_1,g_0)$ be the edge that belonged also in $R$, so on, is an edge of $B_s$.
This implies that there is an alternating path $P_1$ from $z_1$ to $g_p$. 
Preserving to $R_1$ the edge $(z_2,g_0)$ and augmenting this by the path $P_1$ implies that there is a matching saturating every vertex of the subgraph induced by $S\setminus p$ and $(N(S)$.
To show the existence of a matching in the reconstructed feasibility graph in line~\ref{line: reconstruct} it suffices to show that there exists also a perfect matching between $(M_2\setminus S)\cup \{g\}$ and $M_0\setminus \{g\}$.
Recall that $|S|=n-\Delta-n_p+r+k$ and $|N(S)|=n-\Delta-n_p+r$. In this case $k=2$.
This implies that $|M_0|=n-n_p-|N(S)|=\Delta-r$ and $|M_2\setminus S|=n-n_p-|S|=\Delta-r-2$.
Thus, both $(M_2\setminus S)\cup \{g\}$ and $M_0\setminus \{g\}$ have $\Delta-r-1$ vertices.
Moreover, as shown in the proof of Claim~\ref{claimcombined} each item of $M_0$ (so also $g$) may have at most $C=2\Delta-n-n_p-r-k$ other conflict items except those of $S$ and $M_p$.
Thus every items of  $M_0\setminus \{g\}$ is feasible with at least $n-\Delta+n_p+1$ items of $(M_2\setminus S)\cup \{g\}$ since
\begin{align*}
    (\Delta-r-1)-(2\Delta-n-n_p-r-k)=\\
    =n-\Delta+n_p+k-1=n-\Delta+n_p+1
\end{align*}
Let $A=(M_2\setminus S)\cup \{g\}$ and $B=M_0\setminus \{g\}$. are the parts of the feasibility graph. If no matching exists, then for some $X\subseteq  A$ such that $|N(X)|<|X|$.
Since every vertex of $B$ is connected with at least $n-\Delta+n_p+1$ vertices of $A$ then $|B\setminus N(X)|\geq n-\Delta+n_p+2$.
We will show that both cases $g\in X$ and $g\notin X$ lead to a contradiction.

If $g\in X$ then $g$ is in conflict with the items of $B\setminus N(X)$. Since $g$ has at most $C$ conflicts except $S$ and $M_p$, showing that $|B\setminus N(X)|>C$ contradicts this case. Since $|B\setminus N(X)|\geq n-\Delta+n_p+2$ and $k=2$ it suffices
\begin{align*}
     (n-\Delta+n_p+2)-(2\Delta-n-n_p-r-2)&> 0\\
     2n-3\Delta+2n_p+r+4&>0
\end{align*}
which is clear since $\Delta\leq \frac{2n}{3}$ and $r,n_p\geq 0$.

If $g\notin X$ then in $F(A^{\prime}, M_2)$, even if $g$ is feasible with some vertex of $X$, it holds that $|X|\geq |N(X)\cap M_0|$. This contradicts that $(S, N(S))$ is not a maximal Hall violator of maximum deficiency since  $S$ can be extended by the vertices of $X$ and this concludes the proof of this case. 
\end{proof}

\fifthclaim*
\begin{proof}
The fact that Algorithm~\ref{RR-match-v2} re-enter the while loop implies that the criteria of Claim~\ref{claim2} under which Algorithm~\ref{RR-match-v2} are not satisfied.
Let $F(A^*,M_2^*)$ be the feasibility graph when  Algorithm~\ref{RR-match-v2} re-enter the while loop and $F(A^{\prime},M_2)$ in the last iteration before this.
Let $(S^*,N(S^*))$ and $(S,N(S))$ be the Hall violators respectively and $M_0^*$ the respective vertices of $F(A^*,M_2^*)$ not belonging in $M_2^*$ and not in $N(S^*)$ too. 
This Claims says equivalently that there is no vertex $x$ such that $x\in N(S)$ in $F(A^{\prime},M_2)$ and also $x\in M_0^*$.
Assuming that this does not hold, let $X$ be the set of these vertices.
Let $Y$ be the vertices such that $Y\subset S$ and $Y\subset M_2^*\setminus S^*$. 
The vertices of $Y$ can be feasible only with vertices of $X$ from $M_0^*$ because $Y\subset S$. 
It follows that $|Y|<|X|$, otherwise $(S^*,N(S^*))$ would not be a maximal Hall violator of maximum deficiency since it could be extended by $Y$ and $X$.
Since $X\subset M_0^*$ the items of $X$ in $F(A^{\prime},M_2)$ can be feasible only with the items of $Y$ and the items of the feasible pair $p$ allocated to an agent in this iteration.
Let $z_1$ and $z_2$ be the two items that consist of $p$. 
We show that every case of which of $z_1$ and $z_2$ are feasible with some item of $X$ in $F(A^{\prime},M_2)$ lead to a contradiction.

If none of $z_1$ and $z_2$ are feasible with some item of $X$, it contradicts that $(S,N(S))$ is a maximal Hall violator of maximum deficiency, since $|Y|<|X|$ and removing them from $S$ and $(N(S)$ increases the deficiency.

Now, let that only one of $z_1$ and $z_2$ are feasible with some item of $X$, without loss of generality let $z_1$.
It follows that $|Y\cup\{z_1\}|\leq |X|$. 
The strict inequality contradicts for the same reason as in the previous case.
If $|Y\cup\{z_1\}|= |X|$, then, in $R$ the vertices of $Y\cup\{z_1\}$ would be matched with $X$, since $R$ saturates every vertex of $N(S)$. 
This contradicts that $z_1$ and $z_2$ consist of a feasible pair because $z_2$ is not feasible with any item of $X$.

If  both $z_1$ and $z_2$ are feasible with some item of $X$ then it holds that $|Y\cup\{z_1, z_2\}|\leq |X|+1$. 
If $|Y\cup\{z_1, z_2\}|\leq |X|$ contradicts as in first case.
If $|Y\cup\{z_1, z_2\}|= |X|$, in $R$ the vertices of $Y\cup\{z_1\}$ would have been  matched with different vertices of $X$, otherwise $R$ cannot saturate every vertex of $N(S)$, a contradiction.
Now, let $|Y\cup\{z_1, z_2\}|=|X|+1$.
If $z_1$ and $z_2$ that consist the pair $p$ are not matched in $R^{\prime}$ with some vertex of $X$ then $R^{\prime}$ cannot saturate every vertex of $N(S)$.
Otherwise, since the feasible pairs of $S$ are matched with items of $B_s$, it holds that $X\cap B_s\neq \emptyset$.
The vertices of $X$ are feasible with the  $|X|$ vertices matched with $X$ in $R$ and the one of $z_1$ and $z_2$, let $z_2$ without loss of generality, which is unmatched in $R$ and $(z_1,x)\in R$ for $x\in X$. 
These $|X|+1$ vertices belong to $Y\cup\{z_1, z_2\}$ and thus no other vertex belong to this set.
If there is an item $v$ in $S$ not matched with a vertex of $X$ in $R^{\prime}$, then there is an alternating path $P_1$ from $v$ to some vertex $x\in X\cap B_s$ because of Claim~\ref{claimcombined}. In $P_1$ there exists a vertex $v^{\prime}$ such that $v^{\prime}$ is not matched with some vertex of $X$ in $R^{\prime}$. Thus $v^{\prime}\in Y$ and $|Y|\geq |X|+2$ which leads to a contradiction.
The only remaining possible case is that $Y=S$, $X=N(S)$ and $|S|-|N(S)|=1$. By Claim~\ref{claim2} if $A_f$ falls int $C_2$ or $C_3$ the Algorithm~\ref{RR-match-v2} will not re-enter the while loop. $C_1$ is not possible since the item of $A_f$ belong to $N(S)=X$ and in this case we have defined $X$  as the set of vertices of that would be located in $M_0^*$ in the next iteration.
These cases are enough for the induction basis.
If $A_f$ is a feasible pair $p^{\prime}$ i.e falls into Case C4, if $X=N(S)$ and  Algorithm~\ref{RR-match-v2} re-enter the while loop combining with the induction hypothesis implies that $S^*$ contain only one or both items of $S$ and $N(S^*)=\emptyset$. 
This implies that items of $S^*$ are in conflict with $M_0$ ans $N(S)$.
Since $|M_0|\geq n-\Delta+n_p+1$ and $|N(S)|\geq n-\Delta-n_p$ it follows for the conflicts of items of $S^*$ are more that $\Delta$ since
\begin{align*}
    |M_0|+|N(S)|&\geq (n-\Delta+n_p+1)+ (n-\Delta-n_p)\\
    &\geq 2n-2\Delta+1\\
    &\geq 2n-\frac{4n}{3}+1\\
    &=\frac{2n}{3}+1> \Delta.
\end{align*}
This leads to a contradiction and concludes the proof.

\end{proof}

\sixthclaim*
\begin{proof}
Because of Claim~\ref{claim3}, the items of $M_0$ in $F(A^{\prime},M_2)$ are connected in the conflict graph $G$ with every item in $M_p$ where $|M_p|=2n_p$.
If $|S|=\Delta-2n_p$ and $A_f$ can fall into Case C1, then the items belonging to the agents of $N^{\prime}\setminus N(S)$ would be connected in the conflict graph $G$ with $\Delta-2n_p+2n_p+1=\Delta+1$, a contradiction. 
\end{proof}

\end{document}